\documentclass[11pt, a4paper]{article}
\usepackage{amsmath, amsthm, amssymb, url}
\usepackage[margin=1.2in]{geometry}
\usepackage[english]{babel}
\usepackage{multirow}
\usepackage{authblk}
\usepackage{tikz-cd} 
\usepackage{color}
\usepackage{hyperref}
\newcommand{\mms}{\mathrm{mms}}
\newcommand{\Res}{\mathrm{Res}}

\newcommand{\Fix}{\mathrm{Fix}}

\theoremstyle{plain}

\newtheorem{corollary}{Corollary}
\newtheorem{lemma}{Lemma}
\newtheorem{proposition}{Proposition}

\newtheorem{theorem}{Theorem}

\theoremstyle{definition}
\newtheorem{definition}{Definition}
\newtheorem{example}{Example}
\newtheorem{remark}{Remark}

\begin{document}

\title{On the minimal memory set of cellular automata}
\author[1]{Alonso Castillo-Ramirez\footnote{Email: alonso.castillor@academicos.udg.mx}}
\author[2]{Eduardo Veliz-Quintero \footnote{Email: eduardo.veliz9236@alumnos.udg.mx}}
\affil[1]{Centro Universitario de Ciencias Exactas e Ingenier\'ias, Universidad de Guadalajara, M\'exico.}
\affil[2]{Centro Universitario de los Valles, Universidad de Guadalajara, M\'exico.}

\maketitle

\begin{abstract}
For a group $G$ and a finite set $A$, a cellular automaton (CA) is a transformation $\tau : A^G \to A^G$ defined via a finite memory set $S \subseteq G$ and a local map $\mu : A^S \to A$. Although memory sets are not unique, every CA admits a unique \emph{minimal memory set}, which consists on all the essential elements of $S$ that affect the behavior of the local map. In this paper, we study the links between the minimal memory set and the \emph{generating patterns} $\mathcal{P}$ of $\mu$; these are the patterns in $A^S$ that are not fixed when the cellular automaton is applied. In particular, we show that when $\vert S \vert \geq 2$ and $\vert \mathcal{P} \vert$ is not a multiple of $\vert A \vert$, then the minimal memory set must be $S$ itself. Moreover, when $\vert \mathcal{P} \vert = \vert A \vert$, $\vert S \vert \geq 3$, and the restriction of $\mu$ to these patterns is well-behaved, then the minimal memory set must be $S$ or $S \setminus \{s\}$, for some $s \in S \setminus \{e\}$. These are some of the first general theoretical results on the minimal memory set of a cellular automaton.   \\

\end{abstract}

\section{Introduction}

Cellular automata (CA) are transformations of a discrete space defined by a fixed local rule that is applied homogeneously and in parallel in the whole space; they have been used in discrete complex systems modeling, and are relevant in several areas of mathematics, such as symbolic dynamics \cite{LM95} and group theory \cite{CAandG}.         

More formally, let $G$ be a group and let $A$ be a finite set. A function $x : G \to A$ is called a \emph{configuration}, and the set of all configurations is denoted by $A^G$. When $S$ is a finite subset of $G$, a function $p : S \to A$ is called a \emph{pattern} (or a \emph{block}) over $S$, and the set of all patterns over $S$ is denoted by $A^S$. A \emph{cellular automaton} is a transformation $\tau : A^G \to A^G$ defined via a finite subset $S \subseteq G$, called a \emph{memory set} of $\tau$, and a \emph{local map} $\mu : A^S \to A$ such that
\[ \tau(x)(g) = \mu( (g^{-1} \cdot x) \vert_S), \quad \forall x \in A^G, g \in G,\]
where $g^{-1} \cdot x \in A^G$ is the \emph{shift} of $x$ by $g^{-1}$ defined by
\[ (g^{-1} \cdot x)(h) := x(gh), \quad \forall h \in G.  \]
Intuitively, applying $\tau$ to a configuration $x \in A^G$ is the same as applying the local map $\mu : A^S \to A$ homogeneously and in parallel using the shift action of $G$ on $A^G$. Local maps $\mu : A^S \to A$ are also known in the literature as \emph{block maps}. In their classical setting, CA are studied when $G = \mathbb{Z}^d$, for $d \geq 1$, and $A=\{0,1\}$ (e.g., see \cite{Kari}). 

Cellular automata do not have a unique memory set. With the above notation, for any finite superset $S^{\prime} \supseteq S$, we may define $\mu^\prime : A^{S^{\prime}} \to A$ by $\mu^\prime(z) := \mu(z \vert_S)$, for all $z \in A^{S^{\prime}}$, and it follows that $\mu^\prime$ is also a local map that defines $\tau$. Hence, any finite superset of a memory set of $\tau$ is also a memory set of $\tau$. However, cellular automata do have a unique \emph{minimal memory set} (MMS), which is the intersection of all the memory sets admitted by $\tau$ \cite[Sec. 1.5]{CAandG}. Equivalently, the minimal memory set of $\tau$ is the memory set of smallest cardinality admitted by $\tau$, and it consists of all the \emph{essential} elements of $G$ required to define a local defining map for $\tau$. 

So far, there are no general theoretical results about the minimal memory set of a cellular automaton. It is known that it does not behave well with composition: although a memory set of a composition of two cellular automata with memory sets $T$ and $S$ is the product $ST$, the minimal memory set of the composition may be a proper subset of $ST$ \cite[Ex. 1.27]{ExCA}. Similarly, there are no nontrivial results about the minimal memory set of the inverse of an invertible cellular automaton. 

Every cellular automaton $\tau : A^G \to A^G$ admits a memory set $S$ such that $e \in S$, where $e$ is the identity of the group $G$. Let $\mu : A^S \to A$ be the corresponding local map that defines $\tau$. The behavior of $\mu$ may be characterized by the set of patterns $\mathcal{P} \subseteq A^S$ and a function $f : \mathcal{P} \to A$ such that 
\[ \mu(p) = f(p) \neq p(e), \ \forall p \in \mathcal{P}, \quad \text{ and } \quad \mu(z) = z(e), \ \forall z \in A^S \setminus \mathcal{P}. \]
In other words, $\mathcal{P}$ is the set of patterns in which $\mu$ does not act as the projection to $e \in S$; this means, that the cellular automaton $\tau$ fixes a configuration $x \in A^G$ if and only if no element of $\mathcal{P}$ appears as a subpattern in $x$. In such a situation, we say that the pair $(\mathcal{P}, f)$ \emph{generates} the local map $\mu : A^S \to A$. We say that the function $f : \mathcal{P} \to A$ is \emph{well-behaved} if 
\[ \forall p,q \in \mathcal{P},  \quad  p(e) = q(e) \ \Leftrightarrow \ f(p) = f(q). \]
When $A = \{ 0,1 \}$, the function $f : \mathcal{P} \to A$ is always well-behaved because $f(p)$ must be equal to the complement of $p(e)$. 

This approach of characterizing CA by patterns has been proved to be useful in the study of various algebraic and dynamical properties. In \cite{CAIdem}, it was shown that CA generated by a single pattern $p \in A^S$ are often idempotent, like in the case when $p$ is a constant or symmetrical pattern. In \cite{Maas,Coven}, various dynamical properties were examined for the so-called \emph{Coven} CA; in our terminology, these are CA generated by two distinct patterns $p, q \in A^S$ such that $p \vert_{S \setminus \{e\}} = q \vert_{S \setminus \{e\}}$. Remarkably, Coven CA were the first nontrivial class of cellular automata in which the exact computation of the topological entropy was obtained (see \cite[p. 1]{Lind}). 

In this paper, we study the connection between the minimal memory set of a cellular automaton and the set of patterns generating its local map. It has been already shown in \cite[Lemma 1]{CAIdem} that if $\mu : A^S \to A$ is generated by only one pattern, then its minimal memory set is $S$ itself. Here, we establish the following main result.

\begin{theorem}
Let $G$ be a group and let $A$ be a finite set with $\vert A \vert \geq 2$. Let $S \subseteq G$ be a finite subset such that $e \in S$, $\vert S \vert \geq 2$, and let $\mu : A^S \to A$ be a local map. Denote by $\mms(\mu)$ the minimal memory set of the cellular automaton defined by $\mu$. Suppose that the pair $(\mathcal{P}, f)$ generates $\mu$. Then:
\begin{enumerate}
\item If $\vert \mathcal{P} \vert \neq \vert A \vert^{\vert S \vert} - \vert A \vert^{\vert S \vert - 1}$, then $e \in \mms(\mu)$. 

\item If $\vert \mathcal{P} \vert$ is not a multiple of $\vert A \vert$, then $\mms(\mu) = S$. 

\item If $\vert S \vert \geq 3$, $f$ is well-behaved and $\vert \mathcal{P} \vert = \vert A \vert$, then $\mms(\mu) = S$ or $\mms(\mu) = S \setminus \{s\}$, for some $s \in S \setminus \{e\}$.
\end{enumerate}
\end{theorem}

In contrast, if $\vert \mathcal{P} \vert = \vert A \vert^{\vert S \vert} - \vert A \vert^{\vert S \vert - 1}$, there are examples in which the minimal memory set of $\mu$ may be any proper subset of $S$, including the empty set with corresponds to the constant cellular automata (see Example \ref{ex-ECA}). As an application, the previous theorem may be used to improve the brute force algorithm that obtains the minimal memory set of a cellular automaton. 

The structure of this paper is as follows. In Section 2, we set up notation, and present some basic properties on the minimal memory set of a cellular automaton. In Section 3, we present results on the links between the generating set of patterns of a local map $\mu : A^S \to A$ and its minimal memory set, including the proof of Theorem 1.

\section{Basic results}

Let $G$ be a group, and let $A$ be a finite set. For the rest of the paper, we shall assume that $\vert A \vert  \geq 2$, and that $\{0,1\} \subseteq A$. The \emph{configuration space} $A^G$ is the set of all functions of the form $x : G \to A$.

\begin{definition}
The \emph{shift action} of $G$ on $A^G$ is a function $\cdot : G \times A^G \to A^G$ defined by 
\[ (g \cdot x)(h) := x(g^{-1}h), \quad \forall x \in A^G, g, h \in G. \]
\end{definition}

The shift action is indeed a group action in the sense that $e \cdot x = x$, for all $x \in A^G$, where $e$ is the identity element of $G$, and $g \cdot (h \cdot x) = gh \cdot x$, for all $x \in A^G$, $g,h \in G$ (see \cite[p. 2]{CAandG}). 

When $G = \mathbb{Z}$, the configuration space $A^\mathbb{Z}$ may be identified with the set of bi-infinite sequences 
\[ x = \dots x_{-2} x_{-1} x_{0} x_{1} x_{2} \dots \]
for all $x \in A^\mathbb{Z}$, where $x_{k} := x(k) \in A$. The shift action of $\mathbb{Z}$ on $A^\mathbb{Z}$ is equivalent to left and right shifts of the bi-infinite sequences. For example,
\[ 1 \cdot x = \dots x_{-3} x_{-2} x_{-1} x_{0} x_{1} \dots  \] 
For any $k \in \mathbb{Z}$, the bi-infinite sequence $k \cdot x \in A^{\mathbb{Z}}$ is centered at $x_{-k}$. 

\begin{definition}[Def. 1.4.1 in \cite{CAandG}]
A \emph{cellular automaton} is a transformation $\tau : A^G \to A^G$ such that there exists a finite subset $S \subseteq G$, called a \emph{memory set} of $\tau$, and a \emph{local map} $\mu : A^S \to A$, such that 
\[ \tau(x)(g) = \mu( (g^{-1} \cdot x) \vert_S ), \quad \forall x \in A^G, g \in G. \]
\end{definition}

A local map $\mu : A^S \to A$ is also known as a \emph{block map}. We say that a cellular automaton $\tau : A^G \to A^G$ \emph{admits} a memory set $S \subseteq G$ if there exists a local map $\mu : A^S \to A$ that defines $\tau$.  

 The famous Curtis-Hedlund-Lyndon Theorem (see \cite[Theorem 1.8.1]{CAandG}) establishes that a function $\tau : A^G \to A^G$ is a cellular automaton if and only if $\tau$ is \emph{$G$-equivariant} in the sense that $\tau(g \cdot x) = g \cdot \tau(x)$, for all $g \in G$, $x \in A^G$, and $\tau$ is continuous in the \emph{prodiscrete topology} of $A^G$ (which is the product topology of the discrete topology of $A$).

\begin{example}\label{Rule110}
Let $G := \mathbb{Z}$ and $S := \{-1,0,1\} \subseteq G$. The relationship between a cellular automaton $\tau : A^G \to A^G$ with local defining map $\mu : A^S \to A$ is described as follows:
\[ \tau( \dots x_{-1} x_{0} x_{1} \dots ) = \dots \mu(x_{-2},x_{-1},x_{0}) \mu(x_{-1},x_0,x_1) \mu(x_{0}, x_1, x_2) \dots \]
In this setting, it is common to define a local map $\mu : A^S \to A$ via a table that enlists all the elements of $A^S$, which are identified with tuples in $A^3$. For example, 
\[ \begin{tabular}{c|cccccccc}
$z\in A^S$ & $111$ & $110$ & $101$ & $100$ & $011$ & $010$ & $001$ & $000$ \\ \hline
$\mu(z) \in A$ & $0$ & $1$ & $1$ & $0$ & $1$ & $1$ & $1$ & $0$
\end{tabular}\] 
When $A=\{0,1\}$, cellular automata that admit a memory set $S = \{ -1,0,1\} \subseteq \mathbb{Z}$ are known as \emph{elementary cellular automata} (ECA) \cite[Sec. 2.5]{Kari}, and they are labeled with a \emph{Wolfram number}, which is the decimal number corresponding to the second row of the defining table of $\mu : A^S \to A$ considered as a binary number. The Wolfram number of the ECA given by the above table is 110, as the binary number of the second row of the table is $1101110$. 
\end{example}

We write $\mu \sim \nu$ if the local maps $\mu : A^S \to A$ and $\nu : A^T \to A$ define the same cellular automaton. This defines an equivalence relation, and it holds that 
\[ \mu \sim \nu \quad \Leftrightarrow \quad \mu( x \vert_S) = \nu( x \vert_T), \ \forall x \in A^G.  \]

If two local maps $\mu : A^S \to A$ and $\nu : A^T \to A$ define the same cellular automaton $\tau : A^G \to A^G$, then the local map $\lambda : A^{S \cap T} \to A$ defined by 
\[ \lambda(z) := \mu(z \vert_S) = \nu(z \vert_T), \quad \forall z \in A^{S \cap T}, \]
also defines $\tau : A^G \to A^G$ (see \cite[Prop. 1.5.1]{CAandG}). This means that the intersection of any two memory sets for $\tau$ is also a memory set for $\tau$. 

\begin{definition}
The \emph{minimal memory set} (MMS) of a cellular automaton $\tau : A^G \to A^G$, denoted by $\mms(\tau)$ is the intersection of all the memory sets admitted by $\tau$. The minimal memory set of a local map $\mu : A^S \to A$, denoted by $\mms(\mu)$, is the minimal memory set of the cellular automaton defined by $\mu$.  
\end{definition}

Clearly, if $\mu \sim \nu$, then $\mms(\mu)=\mms(\nu)$. It is not hard to show that the MMS of a cellular automaton $\tau : A^G \to A^G$ is the memory set of smallest cardinality admitted by $\tau$, and that $\tau$ admits a memory set $S$ if and only if $\mms(\tau) \subseteq S$ (see \cite[Prop. 1.5.2]{CAandG}).

In the sequel, for each $s \in S$, it will be convenient to consider the function $\Res_s : A^S \to A^{S \setminus \{ s\}}$ defined by
\[ \Res_s(z) := z \vert_{S \setminus \{s\} }, \quad \forall z \in A^S.   \]

\begin{definition}
We say that an element $s \in S$ is \emph{essential} for a local map $\mu : A^S \to A$ if there exist $z,w \in A^S$ such that $\Res_s(z) = \Res_s(w)$ but $\mu(z) \neq \mu(w)$. 
\end{definition}

\begin{proposition}[c.f. Exercise 1.24 in \cite{ExCA}]\label{le-mms}
Let $\mu : A^S \to A$ be a local map. Then, 
\[ \mms(\mu) = \{ s \in S : s \text{ is essential for } \mu \}. \]
\end{proposition}
\begin{proof}
Let $S_0 := \mms(\mu)$. As $S$ is a memory set for the cellular automaton defined by $\mu$, we must have $S_0 \subseteq S$. Suppose that $s \in S$ is not essential for $\mu$. Define $\mu^\prime : A^{S \setminus \{s\}} \to A$ by $\mu^\prime(y) := \mu(\hat{y})$, for all $y \in A^{S \setminus \{s\}}$, where $\hat{y} \in A^S$ is any extension of $y$. The function $\mu^\prime$ is well-defined because $s$ is not essential for $\mu$, so for all $z,w \in A^S$ with $\Res_s(z) = \Res_s(w)$ we have that $\mu(z) = \mu(w)$. Moreover, $\mu \sim \mu^\prime$, so $S_0 \subseteq S \setminus \{s\}$. Hence, $s \not \in S_0$. 

Conversely, suppose there is $s \in S \setminus S_0$. Let $\mu_0 : A^{S_0} \to A$ be the local map associated with $S_0$ that defines the same cellular automaton as $\mu$. If $s$ is essential for $\mu$, there exist $z,w \in A^S$ such that $\Res_s(z) = \Res_s(w)$ but $\mu(z) \neq \mu(w)$. However, $z \vert_{S \setminus \{s\}} = w \vert_{S \setminus \{s\}}$ implies that $\mu_0( z \vert_{S_0}) = \mu_0(w \vert_{S_0})$, as $s \not \in S_0$. This contradicts that $\mu \sim \mu_0$. Therefore, $s$ is not essential for $\mu$.    
\end{proof}

\begin{example}\label{Rule102}
Let $G := \mathbb{Z}$, $A := \{ 0,1 \}$ and $S := \{-1,0,1\}$. Consider the elementary cellular automaton $\tau : A^\mathbb{Z} \to A^\mathbb{Z}$ defined by the local map $\mu :A^S \to A$ described by the following table: 
\[ \begin{tabular}{c|cccccccc}
$z\in A^S$ & $111$ & $110$ & $101$ & $100$ & $011$ & $010$ & $001$ & $000$ \\ \hline
$\mu(z) \in A$ & $0$ & $1$ & $1$ & $0$ & $0$ & $1$ & $1$ & $0$
\end{tabular}\] 
This has Wolfram number 102. In this case, the element $-1 \in S$ is not essential for $\mu$; this may be deduced by crossing out the coordinate corresponding to $-1$ in the tuples of $A^S$, and observing that there are no contradictions in the images of $x_{0}x_1 \in A^{\{0,1\}}$. With this, we may obtain a reduced table corresponding to a local defining map $\mu^\prime : A^{\{0,1\}} \to A$ for $\tau$:
\[ \begin{tabular}{c|cccc}
$z\in A^{\{0,1 \}}$ & $11$ & $10$ & $01$ & $00$  \\ \hline
$\mu^\prime(z) \in A$ & $0$ & $1$ & $1$ & $0$
\end{tabular}\] 
Now, we may check that $0$ and $1$ are both essential for $\mu^\prime$ (for example, $\Res_0(11)=\Res_0(01)$ and $\mu(11)=0 \neq 1 = \mu(01)$), so
\[ \mms(\mu) = \mms(\mu^{\prime}) = \{0,1\}. \] 
\end{example}

\begin{example}
In the case of the local map $\mu : A^S \to A$ with Wolfram number 110 given by Example \ref{Rule110}, we may check that all the elements of $S=\{-1,0,1\}$ are essential for $\mu$, so $\mms(\mu) = S$.  
\end{example}


\section{Cellular automata generated by patterns} 

For the rest of the paper, assume that $S$ is a finite subset of $G$ such that $e \in S$. 

\begin{definition}\label{CA-pattern}
We say that a pair $(\mathcal{P}, f)$ \emph{generates} a local map $\mu : A^S \to A$ if 
\[ \mathcal{P} := \{ z \in A^S : \mu(z) \neq z(e) \},  \]
and $f : \mathcal{P} \to A$ is the restriction of $\mu$ to $\mathcal{P}$. 
\end{definition}

In other words, $\mathcal{P}$ is the set of patterns on which $\mu : A^S \to A$ does not act as the projection to $e$. Observe that if $(\mathcal{P}, f)$ generates $\mu$, then
\[ \mu(z) = \begin{cases}
f(z) & \text{ if } z \in \mathcal{P} \\
z(e) & \text{ if } z \not \in \mathcal{P}
\end{cases},  \quad \forall z \in A^S. \]

Moreover, $\mu : A^S \to A$ is equal to the projection to $e$ (which means that the cellular automaton defined by $\mu$ is the identity function) if and only if $\mathcal{P} = \emptyset$. It follows by definition that every local map $ \mu :A^S \to A$ has a unique generating pair $(\mathcal{P},f)$. 

\begin{remark}\label{remark1}
If $A= \{0,1\}$, for any $\mathcal{P} \subseteq A^S$, there is a unique choice for the function $f : \mathcal{P} \to A$ because of the condition that $f(p) =\mu(p) \neq p(e)$, for all $p \in \mathcal{P}$. Explicitly, $f$ must be defined by $f(p):= p(e)^c$, where $p(e)^c$ denotes the complement of $p(e)$. Hence, in this situation, we simply say that the set of patterns $\mathcal{P} \subseteq A^S$ generates $\mu : A^S \to A$. 
\end{remark}

\begin{example}
Let $G := \mathbb{Z}$, $S := \{-1,0,1\}$ and $A:=\{0,1\}$.
\begin{enumerate}
\item Let $\mu : A^S \to A$ be the local map with Wolfram number 110 given by Example \ref{Rule110}. Then $\mu$ is generated by the set of patterns $\mathcal{P} = \{111, 101, 001\}$. 
\item Let $\mu : A^S \to A$ and $\mu^\prime : A^{\{0,1\}} \to A$ be the local maps with Wolfram number 102 given by Example \ref{Rule102}. Then $\mu$ is generated by the set of patterns $\{ 111, 101, 011, 001 \}$ and $\mu^\prime$ is generated by the set of patterns $\{ 11, 01\}$.
 \end{enumerate}
\end{example}

A fundamental object in symbolic dynamics is a \emph{subshift}, which may be defined as a closed (in the prodiscrete topology) \emph{$G$-equivariant} subset $X$ of $A^G$, in the sense that $g \cdot x \in X$ for all $g \in G$, $x \in X$. Equivalently, any subshift $X \subseteq A^G$ may be defined via a (possibly infinite) set of \emph{forbidden patterns} (see Ex. 1.39 and 1.47 in \cite{ExCA}). In our setting, for $\mathcal{P} \subseteq A^S$, we shall consider the subshift $X_{\mathcal{P}} \subseteq A^G$ defined by the forbidden patterns $\mathcal{P}$:
\[  X_{\mathcal{P}} := \{ x \in A^G : (g \cdot x)\vert_S \not\in \mathcal{P}, \forall g \in G \}. \]
This is a \emph{subshift of finite type}, since the set of forbidden patterns is a finite set. 

\begin{lemma}[c.f. Exercise 1.61 in \cite{ExCA}]
Let $\tau : A^G \to A^G$ be a cellular automaton with memory set $S \subseteq G$, with $e \in S$, and local defining map $\mu : A^S \to A$. If $(\mathcal{P},f)$ generates $\mu$, then
\[ \Fix(\tau) := \{ x \in A^G : \tau(x) = x \} = X_{\mathcal{P}}. \] 
\end{lemma}
\begin{proof}
Let $x \in X_{\mathcal{P}}$. Then, $(g \cdot x)\vert_S \not\in \mathcal{P}$, for all $g \in G$, so it follows that
\[ \tau(x)(g) = \mu( (g^{-1} \cdot x) \vert_S) = (g^{-1} \cdot x)(e) = x(g), \quad \forall g \in G.   \]
Therefore, $x \in \Fix(\tau)$. Conversely, suppose that $x \not \in  X_{\mathcal{P}}$, so there exists $g \in G$ such that $(g \cdot x)\vert_S \in \mathcal{P}$. Then,
\[ \tau(x)(g) = \mu( (g^{-1} \cdot x) \vert_S) = f( g^{-1} \cdot x) \neq ( g^{-1} \cdot x)(e) = x(g).  \]
This shows that $\tau(x) \neq x$, so $x \not \in \Fix(\tau)$. 
\end{proof}

We now turn our attention to the minimal memory set of a local map $\mu : A^S \to A$ generated by a pair $(\mathcal{P},f)$. In the following example, we examine in detail the links between the generating patterns and minimal memory sets for elementary cellular automata. 

\begin{example}\label{ex-ECA}
Let $A:=\{0,1\}$ and $S: = \{-1,0,1\} \subseteq \mathbb{Z}$. Table \ref{mms}, which was obtained by direct computations, shows the sizes of the minimal memory sets of local maps $\mu : A^S \to A$ according to the sizes of their generating set of patterns $\mathcal{P} \subseteq A^S$.
\begin{table}[!h]\centering
\setlength\tabcolsep{.7em}
\caption{Generating patterns and minimal memory set of ECA.}\label{mms}
\begin{tabular}{|c|c|c|c|c|c|c|c|c|c|}\hline
 $\vert \mathcal{P} \vert$  &  $0$  &  $1$  &  $2$  &   $3$  &  $4$  &  $5$  &  $6$  &  $7$  &   $8$  \\ \hline
 \text{No. local maps}  &  $1$  &   $8$  &  $28$   &  $56$  & $70$  &  $56$  &  $28$  &  $8$  &   $1$  \\ \hline
  \text{Size of MMS}   &  $1$  &   $3$  &  $3 \text{ or } 2$  &  $3$  &  $0,1,2, \text{ or } 3  $   &  $3$  &  $3 \text{ or } 2 $  &  $3$ & $1$  \\ \hline 
\end{tabular}
\end{table}

Since $\vert A^S \vert = 8$, there are $\binom{8}{k}$ local maps generated by $k =: \vert \mathcal{P} \vert$ different patterns. When $k=0$, cellular automaton must be the identity, while when $k=8$, the cellular automaton must be the rule that exchanges $0$'s and $1$'s (ECA 51). The richest variety of minimal memory sets appears when $k=4$, including the constant cellular automata (ECA 0 and 255) whose minimal memory set is the empty set $\emptyset$. The symmetry that appears in the possible sizes of MMS in Table \ref{mms} may be explained by Proposition \ref{le-complement} (2). 
\end{example}

The minimal memory set of a local map $\mu : A^S \to A$ may be more easily characterized in terms of a generating pair $(\mathcal{P}, f)$ when the function $f : \mathcal{P} \to A$ only depends on the projection to $e$ and acts as a permutation of $A$. Hence, we introduce the following definition.

\begin{definition}
Suppose that $(\mathcal{P}, f)$ generates a local map $\mu : A^S \to A$. We say that $f : \mathcal{P} \to A$ is \emph{well-behaved} if, for all $p,q \in \mathcal{P}$, $p(e) = q(e)$ if and only if $f(p) = f(q)$.
\end{definition}

\begin{remark}
If $A = \{ 0,1\}$, then $f : \mathcal{P} \to A $ is always well-behaved because, as explained in Remark \ref{remark1}, $f(p) = p(e)^c$. Clearly, $p(e) = q(e)$ if and only if $p(e)^c = q(e)^c$. 
\end{remark}

In the following lemmas, we shall analyze when $s \in S$ is essential for a local map $\mu : A^S \to A$ generated by a pair $(\mathcal{P},f)$. The two distinctive cases will be when $s \neq e$ and $s = e$.  

Recall that for $s \in S$, we define $\Res_s : A^S \to A^{S \setminus \{ s\}}$ by $\Res_s(z) = z \vert_{S \setminus \{s \}}$. For $\mathcal{P} \subseteq A^S$, denote $\mathcal{P}^c := A^S \setminus \mathcal{P}$. 

\begin{lemma}\label{le-1}
Suppose that $(\mathcal{P},f)$ generates a local map $\mu : A^S \to A$. Let $s \in S \setminus \{e\}$. 
\begin{enumerate}
\item If there exist $p \in \mathcal{P}$ and $z \in \mathcal{P}^c$ such that $\Res_s(p) = \Res_s(z)$, then $s$ is essential for $\mu$.

\item If $s$ is essential for $\mu$ and $f$ is well-behaved, then there exist $p \in \mathcal{P}$ and $z \in \mathcal{P}^c$ such that $\Res_s(p) = \Res_s(z)$.
\end{enumerate}
\end{lemma}
\begin{proof}
For point (1), observe that $z(e) = p(e)$ because $\Res_s(p) = \Res_s(z)$ and $s \neq e$. Then,
\[ \mu(z) = z(e) = p(e) \neq f(p) = \mu(p).   \]
It follows that $s$ is essential for $\mu$.  

For point (2), suppose that for all $p \in \mathcal{P}$ and $z \in A^S$ such that $\Res_s(p) = \Res_s(z)$, we have $z \in \mathcal{P}$. Take arbitrary $z_1, z_2 \in A^S$ such that $\Res_s(z_1) = \Res_s(z_2)$. We have two cases:
\begin{itemize}
\item \textbf{Case $z_1 \in \mathcal{P}$}: By assumption, we must have $z_2 \in \mathcal{P}$. As $f$ is well-behaved and $z_1(e) = z_2(e)$, then $f(z_1) = f(z_2)$. Then, 
\[ \mu(z_1) = f(z_1) = f(z_2) = \mu(z_2). \]

\item \textbf{Case $z_1 \not\in \mathcal{P}$}: By assumption, we must have $z_2 \not \in \mathcal{P}$. Then,
\[ \mu(z_1) = z_1(e) = z_2(e) = \mu(z_2).  \] 
\end{itemize}
This contradicts that $s$ is essential for $\mu$. 
\end{proof}

\begin{proposition}\label{le-2}
Suppose that $(\mathcal{P},f)$ generates a local map $\mu : A^S \to A$ and that $\vert \mathcal{P} \vert$ is not a multiple of $\vert A \vert$. Then every $s \in S \setminus \{e\}$ is essential for $\mu$.
\end{proposition}
\begin{proof}
Fix $s \in S \setminus \{e\}$. By Lemma \ref{le-1} (1), it is enough if we show that there exist $p \in \mathcal{P}$ and $z \in \mathcal{P}^c$ such that $\Res_s(p) = \Res_s(z)$. If no such pair exists, it means that $\mathcal{P}$ may be written as a partition of preimages under $\Res_s : A^S \to A^{S \setminus \{ s\}}$:
\[ \mathcal{P} = \bigsqcup_{i=1}^n \Res_s^{-1}(y_i),  \]
for some $y_i \in A^{S \setminus \{ s\}}$, $i \in \{1,2 \dots, n\}$. However, 
\[ \vert \Res_s^{-1}(y_i) \vert = \vert A \vert, \quad \forall i \in \{1,2 \dots, n\}.  \]
Hence, $\vert \mathcal{P} \vert = n \vert A \vert$ is a multiple of $\vert A \vert$, which contradicts the hypothesis. 
\end{proof}

Now we shall try to determine when $e \in S$ is essential for a local map $\mu : A^S \to A$ generated by $(\mathcal{P},f)$. 

\begin{lemma}\label{le-3}
Suppose that $(\mathcal{P},f)$ generates a local map $\mu : A^S \to A$.
\begin{enumerate}
 \item If there exist $z, w \in \mathcal{P}^c$ such that $z \neq w$ and $\Res_e(z) =\Res_e(w)$, then $e \in S$ is essential for $\mu$. 
 \item Suppose that $f$ is well-behaved and that there exist $p, q \in \mathcal{P}$ such that $p \neq q$ and $\Res_e(p) =\Res_e(q)$. Then, $e \in S$ is essential for $\mu$. 
 \end{enumerate}
\end{lemma}
\begin{proof}
For part (1), observe that $z(e) \neq w(e)$ because $z \neq w$ and $\Res_e(z) =\Res_e(w)$. Since $z, w \in \mathcal{P}^c$, we have
\[ \mu(z) = z(e) \neq w(e) = \mu(w).  \]
It follows that $e$ is essential for $\mu$. 

For part (2), we also have $p(e) \neq q(e)$. As $f$ is well-behaved, then $f(p) \neq f(q)$, so
\[ \mu(p) = f(p) \neq f(q) = \mu(q). \]
This shows that $e$ is essential for $\mu$.
\end{proof}

\begin{corollary}\label{cor-A3}
Suppose that $(\mathcal{P},f)$ generates a local map $\mu : A^S \to A$, with $f$ well-behaved, $\vert S \vert \geq 2$ and $\vert A \vert \geq 3$. Then $e$ is essential for $\mu$. 
\end{corollary}
\begin{proof}
First, $A^{S \setminus \{e\}} \neq \emptyset$ because $\vert S \vert \geq 2$. For any $y \in A^{S \setminus \{e\}}$, the set $\Res^{-1}_e(y)$ has size $\vert A \vert \geq 3$. Hence, we must have that either $\vert \Res^{-1}_e(y) \cap \mathcal{P} \vert \geq 2$ or $\vert \Res^{-1}_e(y) \cap \mathcal{P}^c \vert \geq 2$. It follows from Lemma \ref{le-3} that $e$ is essential for $\mu$. 
\end{proof}

\begin{corollary}\label{cor-e}
Suppose that $(\mathcal{P},f)$ generates a local map $\mu : A^S \to A$, and that $A = \{0,1\}$. Then, $e \in S$ is essential for $\mu$ if and only if there exist $z, w \in \mathcal{P}^c$, or $z, w \in \mathcal{P}$, such that $z \neq w$ and $\Res_e(z) =\Res_e(w)$. 
\end{corollary}
\begin{proof}
The converse implication follows by Lemma \ref{le-3}. Suppose that $e$ is essential for $\mu$. By definition, there exist $z, w \in A^S$ such that $\Res_e(z) = \Res_e(w)$ and $\mu(z) \neq \mu(w)$. Note that $z(e) \neq w(e)$ (as otherwise, $z=w$). We will show that we must have that either $z, w \in \mathcal{P}^c$ or $z, w \in \mathcal{P}$. For a contradiction, suppose that $z \in \mathcal{P}$ and $w \in \mathcal{P}^c$. Then, $\mu(w) = w(e)$ and, by Remark \ref{remark1}, we have $\mu(z) = f(z) = z(e)^c$. However, since $z(e) \neq w(e)$ and $A=\{0,1\}$, we must have that $w(e) = z(e)^c$. This contradicts that $\mu(z) \neq \mu(w)$. 
\end{proof}

\begin{proposition}\label{le-complement}
Let $\mathcal{P} \subseteq A^S$. Let $f : \mathcal{P} \to A$ and $g : \mathcal{P}^c \to A$ be two well-behaved functions. Suppose that $(\mathcal{P},f)$ generates $\mu : A^S \to A$ and that $(\mathcal{P}^c, g)$ generates $\mu^\prime : A^S \to A$. 
\begin{enumerate}
\item $s \in S \setminus \{e\}$ is essential for $\mu$ if and only if $s \in S \setminus \{e\}$ is essential for $\mu^\prime$. 
\item If $A=\{0,1\}$, then $\mms(\mu) = \mms(\mu^\prime)$.
\end{enumerate}
\end{proposition}
\begin{proof}
For part (1), we use Lemma \ref{le-1}. It follows that $s \in S \setminus \{e\}$ is essential for $\mu$ if and only if there is $p \in \mathcal{P}$ and $z \in \mathcal{P}^c$ such that $\Res_s(p) = \Res_s(z)$, which holds if and only if $s$ is essential for $\mu^\prime$. Part (2) follows by Corollary \ref{cor-e}.
\end{proof}

\begin{remark}
When $A=\{0,1\}$, the local map $\mu^\prime : A^S \to A$ is different from what is known in the literature as the \emph{complementary rule}, which is induced by the group-theoretic conjugation by the invertible cellular automaton that exchange $0$'s and $1$'s (ECA 51). For example, the complementary rule of the ECA 110 is the ECA 137; however, the cellular automaton generated by the complementary patterns of the ones that generate ECA 110 is the ECA 145 (which is the result of only composing on one side by ECA 51).
\end{remark}

\begin{lemma}\label{cor-1}
Suppose that $(\mathcal{P},f)$ generates a local map $\mu : A^S \to A$, with $\vert S \vert \geq 2$. If $e$ is not essential for $\mu$, then 
\[ \vert \mathcal{P} \vert = \vert A \vert^{\vert S \vert} - \vert A \vert^{\vert S \vert-1}.  \]
\end{lemma}
\begin{proof}
As $e$ is not essential for $\mu$, by Lemma \ref{le-3} (1) we have that all $z,w \in \mathcal{P}^c$, $z \neq w$, satisfy $\Res_e(z) \neq \Res_e(w)$. This means that $\Res_e : \mathcal{P}^c \to  A^{S \setminus \{ e\}}$ is an injective function, with $A^{S \setminus \{e\}} \neq \emptyset$ because $\vert S \vert \geq 2$. Hence $\vert \mathcal{P}^c \vert \leq \vert A \vert^{\vert S \vert-1}$, which is equivalent to $\vert \mathcal{P} \vert \geq \vert A \vert^{\vert S \vert} - \vert A \vert^{\vert S \vert-1}$. 

On the other hand, that $e$ is not essential for $\mu$, implies that $\mu$ is constant on $\Res_e^{-1}(y)$ for all $y \in A^{S \setminus \{e\}}$. Hence, for all $y \in A^{S \setminus \{e\}}$, there exists a unique $\hat{y} \in \mathcal{P}^c$ such that $\hat{y} \in \Res_e^{-1}(y)$ (namely, if $\mu(w) = a \in A$ for all $w \in \Res_e^{-1}(y)$, let $\hat{y} \in \Res_e^{-1}(y)$ be such that $\hat{y}(e) = a$; it follows by the definition of $\mathcal{P}$ that $\mathcal{P}^c \cap \Res_e^{-1}(y) = \{ \hat{y}  \} $). This implies that there is an injective function $A^{S \setminus \{e\}} \to \mathcal{P}^c$ given by $y \mapsto \hat{y}$, so $\vert A \vert^{\vert S \vert - 1} \leq \vert \mathcal{P}^c \vert$. Therefore,  $\vert \mathcal{P} \vert \leq \vert A \vert^{\vert S \vert} - \vert A \vert^{\vert S \vert-1}$, and the result follows.  
\end{proof}

\begin{corollary}\label{cor-2}
Suppose that $(\mathcal{P},f)$ generates a local map $\mu : A^S \to A$, where $\vert S \vert \geq 2$ and $\vert \mathcal{P} \vert$ is not a multiple of $\vert A \vert$. Then, 
\[ \mms(\tau) = S. \]  
\end{corollary}
\begin{proof}
By Proposition \ref{le-2}, every $s \in S \setminus \{e\}$ is essential for $\mu$. It follows by Lemma \ref{cor-1} that $e \in S$ is also essential for $\mu$, because $\vert \mathcal{P} \vert \neq \vert A \vert^{\vert S \vert} - \vert A \vert^{\vert S \vert-1}$. 
\end{proof}

The following result is Lemma 1 in \cite{CAIdem}, which is an immediate consequence of Corollary \ref{cor-2}.

\begin{corollary}\label{le-one-p}
Let $S \subseteq G$ be a finite subset such that $e \in S$ and $\vert S \vert \geq 2$. Let $\mu : A^S \to A$ be a local map generated by $(\mathcal{P}, f)$ with $\vert \mathcal{P} \vert = 1$. Then, 
\[ \mms(\mu) = S. \] 
\end{corollary}

\begin{proposition}\label{prop-last}
Let $S \subseteq G$ be a finite subset such that $e \in S$ and $\vert S \vert \geq 3$. Suppose that $(\mathcal{P},f)$ generates a local map $\mu : A^S \to A$, with $f$ well-behaved and $\vert \mathcal{P} \vert = \vert A \vert$. Then, $\mms(\mu) = S \setminus \{s\}$, for some $s \in S \setminus \{e\}$, or $\mms(\mu) = S$.   
\end{proposition}
\begin{proof}
We divide the proof in two cases.
\begin{itemize}
\item \textbf{Case 1:} There exists $s \in S \setminus \{e\}$ such that for all $p, q \in \mathcal{P}$ we have $\Res_s(p) = \Res_s(q)$. Since $\vert \mathcal{P} \vert = \vert A \vert$, there is no $z \in \mathcal{P}^c$ such that $\Res_s(z) = \Res_s(p)$ for some $p \in \mathcal{P}$. By Lemma \ref{le-1} (2), $s$ is not essential for $\mu$. Now we can consider $\mu^\prime : A^{S \setminus \{s\}} \to A$ defined by $\mu^\prime(z) := \mu(\hat{z})$, for all $z \in A^{S \setminus \{s\}}$, where $\hat{z} \in A^S$ is any extension of $z$ (this is well-defined because $s$ is not essential for $\mu$). Then $\mu^\prime$ is generated by $(\{ \Res_s(p) \}, f^\prime)$, for any $p \in \mathcal{P}$, where $f^\prime(\Res_s(p) ) = f(p)$. Since $\mu \sim \mu^\prime$, and $\mu^\prime$ is generated by a single pattern over $A^{S \setminus \{s\}}$ with $\vert S \setminus \{s\}\vert \geq 2$, it follows from Corollary \ref{le-one-p} that
\[\mms(\mu) = \mms(\mu^\prime) = S \setminus \{s\}. \]

\item \textbf{Case 2:} For all $s \in S \setminus \{e\}$ there exist $p, q \in \mathcal{P}$ such that $\Res_s(p) \neq \Res_s(q)$. Observe that $\Res_s^{-1}(\Res_s(p))$ has size $\vert A \vert = \vert \mathcal{P} \vert$ and $q \in \mathcal{P} \cap \Res_s^{-1}(\Res_s(p))^c$, so it is not possible that $\Res_s^{-1}(\Res_s(p)) = \mathcal{P}$. Therefore, there exists $z \in \mathcal{P}^c$ such that $z \in \Res_s^{-1}(\Res_s(p))$, which means that $\Res_s(z) = \Res_s(p)$. It follows from	 Lemma \ref{le-1} (1) that $s$ is essential for $\mu$. Now, since $\vert S \vert \geq 3$, then $\vert \mathcal{P} \vert = \vert A \vert < \vert A \vert^{\vert S \vert} - \vert A \vert^{\vert S \vert-1}$. Hence, it follows from Lemma \ref{cor-1} that $e \in S$ is essential for $\mu$. 
\end{itemize} 
\end{proof}

The proof of Theorem 1 follows from Lemma \ref{cor-1}, Corollary \ref{cor-2}, and Proposition \ref{prop-last}.


\section*{Acknowledgments}

The second author was supported by CONAHCYT \emph{Becas nacionales para estudios de posgrado}, Government of Mexico. We sincerely thank all the comments and suggestions made by the anonymous reviewers of this paper, especially for fixing Case 2 in the proof of Proposition \ref{prop-last}.


\end{document}